\documentclass[12pt]{article}
\usepackage{euler}
\usepackage{amssymb,amsmath,amsthm,amscd,latexsym}
\usepackage{tikz-cd}
\usepackage{amsfonts}
\usepackage[mathscr]{eucal}
\usepackage{color}

\newtheorem{theorem}{Theorem}[section]

\newtheorem{lem}[theorem]{Lemma}

\newtheorem{example}{Example}
\newtheorem{defi}{Definition}

\font\msbm=msbm10 at 12pt

\newcommand{\FF}{\mbox{\msbm F}}

\def\vv{\mathbf{v}}

\def\vw{\mathbf{w}}
\def\vc{\mathbf{c}}

\def\1v{\mathbf{1}}
\def\0v{\mathbf{0}}

\begin{document}
\title{ Group LCD and Group Reversible LCD Codes}
\author{
Steven T. Dougherty \\
Department of Mathematics \\
University of Scranton \\
Scranton, PA 18510 \\
USA \\
Joe Gildea\\
Adrian Korban \\ 
Adam M. Roberts\\
University of Chester \\
Chester, UK \\
}

\maketitle

\begin{abstract}
In this paper, we give a new method for constructing LCD codes. We employ group rings and a  well known map that sends group ring elements to a subring of the $n \times n$ matrices to obtain LCD codes. Our construction method guarantees that our LCD codes are also group codes, namely, the codes are ideals in a group ring. We show that with a certain condition on the group ring element $v,$ one can construct non-trivial group LCD codes. Moreover, we also show that by adding more constraints on the group ring element $v,$ one can construct group LCD codes that are reversible. We present many examples of binary group LCD codes of which some are optimal and group reversible LCD codes with different parameters.
 
\end{abstract} 


\section{Introduction} 

In this paper, we shall study Linear Complementary Dual codes (LCD codes), which are linear codes that have a trivial intersection with their orthogonal.  LCD codes were first introduced by Massey in
\cite{Massey} and were used to give an optimum linear coding solution for the two
user binary adder channel.  These codes are also used in counter measures for
passive and active side channel analyses on embedded crypto-systems.  For a detailed description of this application please see \cite{CarGui}. LCD codes are asymptotically good, it is shown in \cite{Sendrier1} that they meet the asymptotic Gilbert-Varshamov bound. One of the main goals is  to construct LCD codes over finite fields that also have good error correcting properties. A result of Carlet et al. in \cite{CarletMesnager} gives that over any finite field of order $q>3,$ the existence of an $[n,k,d]_q$ linear code implies the existence of an $[n,k,d]_q$ LCD code. In general, the limits to error correction may be more restrictive for binary and ternary LCD codes than for linear codes, as LCD codes do satisfy extra conditions.

Finding methods for constructing LCD codes, classifying LCD $[n,k]$ codes and determining the largest minimum weight among all LCD $[n,k]$ codes are open problems in coding theory. Many researchers have employed different techniques to construct LCD codes. In \cite{Harada1}, the authors consider $k$-covers to construct binary LCD codes with and they give a complete classification of these codes with the largest minimum weight for $1\leq k \leq n \leq 16.$ In the paper by D. Crnkovic et al. (\cite{Crnkovic1}), LCD codes are constructed from weighing matrices and a connection between LCD codes and self-dual codes is drawn. In \cite{Shi1}, M. Shi et al. construct LCD codes from tridiagonal Toeplitz matrices. In this paper, the authors construct optimal and quasi-optimal examples of binary and ternary LCD codes. Recently in \cite{Araya1}, the authors determine the largest minimum weight of LCD $[n,4]$ codes for $n \equiv 2,3,5,6,9,10,13 \pmod {15}$ and the largest minimum weight of LCD $[n,5]$ codes for $n \equiv 3,4,5,7,11,19,20,22,26 \pmod {31}.$ In another recent paper by S. Bouyuklieva (\cite{Bouyuklieva}), optimal binary LCD $[n,k]$ codes are given for $k \leq 32$ and $n \leq 40.$

In this work, we present a new method for constructing LCD codes over any finite commutative ring. We employ group rings and a very well known matrix construction to study and construct LCD codes. In particular, we construct LCD codes that are group codes generated from group ring elements. We show that for a specific choice of a group element $v$ in the group ring $RG,$ one can obtain LCD codes. We specifically show that if we let the coefficient of the group element $g_i$ be the same as the coefficient of the group element $g_i^{-1}$ in the group ring element $v,$ then one can construct a non-trivial group LCD code over any finite commutative ring $R.$ Moreover, by further restricting the choice of the group in the group ring $RG,$ and by fixing the listing of the elements of that group, we show that one can construct group reversible LCD codes. We give many examples of group LCD and group reversible LCD codes of different lengths.

The paper is organized as follows. In Section~2, we give the standard definitions and notations for linear codes, LCD codes, the alphabet that we use in this work, group rings and group codes. In Section~3, we present our main method for constructing group LCD codes. We also give an example of how our method works. Next, in Section~4, we describe generator matrices for various groups in the the group ring $RG$ with the condition for obtaining group LCD codes with. In Section~5, we further amend our construction method so that one can construct group reversible LCD codes. Then, in Section~6, we give many examples of group LCD codes and group reversible LCD codes. We finish with concluding remarks and directions for possible future research.

\section{Definitions and Notations}

\subsection{Linear codes and LCD codes}

Our main aim in this paper is to construct binary LCD codes.  However, we shall make use of linear codes over rings to construct such codes.  Therefore, we shall make a very broad definition of linear codes.  That is, we assume that the alphabet of a code can be any finite commutative ring.  Note that we assume that a ring has a multiplicative unity.   
A code $C$, of length $n$, over a
finite ring $R$ is a subset of $R^n.$   If the code is a submodule then
the code is said to be linear. We attach to the ambient space  the
standard inner-product, namely $[\vv,\vw ] = \sum v_i w_i.$ The
orthogonal is defined by $C^\perp = \{ \vv \in R^n \ | \
[\vv,\vw]=0,\ \forall \vw \in C \}.$   If the ring is Frobenius, then we have that for linear codes $|C| |C^\perp| = |R|^n$.  For a complete description of codes over rings and this fact about their cardinality, see \cite{Doughertybook}.

The automorphism group of a code $C$, denoted by $Aut(G)$, is the group of all permutations of the coordinates of the code that fix the code.

An LCD code satisfies  $C \cap C^\perp = \{  {\bf 0} \}$. 
The following are well known results.   Let $G$ be a generator matrix for a code over a finite field. If
$GG^T = I_n$ then $G$ generates an LCD code.
Let $G$ be a generator matrix for a code over a field.  Then $det(GG^T) \neq 0$ if and only if $G$ generates an LCD code.

There is a particular family of finite commutative Frobenius rings that we shall use later in our work to prove some results with, which is a commutative ring with characteristic 2.  
The ring $R_k$ is defined as $R_k = \FF_2[u_1,u_2,\dots,u_k] /
\langle  u_i^2, u_iu_j - u_j u_i \rangle .$ We have that  $R_k$ has size
$|R_k| = 2^{2^k}$ and it is a non-chain ring which has
characteristic $2$ with
    maximal ideal ${\bf M}=\langle   u_1, u_2, \ldots, u_k \rangle$
    and $Soc(R_k)=\langle u_1u_2 \cdots u_k \rangle.$

A linear Gray map from $R_k$ to $\FF_2^{2^k}$ is constructed as follows.  
Let $\phi_1$ be the map defined on $R_1$, namely $\phi_1(a+bu) = (b,
a+b)$.  Then let ${{c}}\in R.$ We can write ${{c}}={{c}}_1+u_k
{{c}}_2$ where ${{c}}_1,{{c}}_2$ are elements of the ring $R_{k-1}$
of order $2^{2^{k-1}}$,  then we define
\begin{equation} \phi_k({{c}})=(\phi_{k-1}({{c}}_2),\phi_{k-1}({{c}}_1)+\phi_{k-1}({{c}}_2)).\end{equation}
 The map $\phi_k$ is a weight preserving map which we then  expand  coordinatewise to $R^n$.

In  \cite{Dougherty5} it is shown that the map $\phi_k:R_k \rightarrow \FF_2^{2^k}$ is a linear bijection and we have that $\phi(C^\perp) = \phi(C)^\perp.$

It follows immediately that if 
$C$ is an LCD code of length $n$ over $R_k$ then $\phi(C)$ is a
binary LCD code of length $2^kn.$  For a complete description of codes over $R_k$, see \cite{Dougherty5}, \cite{Dougherty6} and \cite{Dougherty8}.

\subsection{Special Matrices and Group Rings} 

In this section, we recall the definitions of some special matrices that we use later in this work. We also give the basic definitions of group rings.

A circulant matrix is one where each row is shifted one element to the right relative to the preceding row. We label the circulant matrix as $circ(\alpha_1,\alpha_2,\dots , \alpha_n),$ where $\alpha_i$ are the ring elements appearing in the first row. A reverse-circulant matrix is one where each row is shifted one element to the left relative to the preceding row. We label the reverse-circulant matrix as $revcirc(\alpha_1,\alpha_2,\dots,\alpha_n),$ where $\alpha_i$ are ring elements appearing in the first row. The transpose of a matrix $A,$ denoted by $A^T,$ is a matrix whose rows are the columns of $A,$ i.e., $(A^T)_{ij}=A_{ji}.$

Our interest is to construct codes that are not only LCD but also  that are ideals inside of a group ring, where the ring is the alphabet of the code.  We shall now give the definitions for group rings.
We let $G$ be a finite group or order $n$,  the group ring $RG$ consists of $\sum_{i=1}^n \alpha_i g_i$, $\alpha_i \in R$, $g_i \in G.$  

Addition in the group ring is done by coordinate addition, namely $$\sum_{i=1}^n \alpha_i g_i + \sum_{i=1}^n \beta_i g_i = 
\sum_{i=1}^n (\alpha_i + \beta_i ) g_i.$$ The product of two elements in  a group ring  is given by 
$$\left(\sum_{i=1}^n \alpha_i g_i\right)\left( \sum_{j=1}^n \beta_j g_j\right)  = \sum_{i,j} \alpha_i \beta_j g_i g_j .$$
It follows immediately that the coefficient of $g_k$ in the product is $ \sum_{g_i g_j = g_k } \alpha_i \beta_j .$

When the ring is a field the group ring is generally referred to as a group algebra.  
While group rings are defined for groups and rings of arbitrary cardinality, we shall only consider when both the group and the ring are finite. Moreover, we shall assume that the ring $R$ is commutative but we let $G$ be an arbitrary group and make no assumption about its commutativity.  
Throughout this paper we use $e$ to refer to the identity element of the  group $G$.

\subsection{Code construction} 

The following map was used in \cite{us} to study group codes over a finite commutative Frobenius ring.
   
Let $R$ be a finite commutative Frobenius ring and let $G = \{ g_1,g_2,\dots,g_n \}$ be a group of order $n$.  Let $v = \alpha_{g_1} g_1 + \alpha_{g_2} g_2 + \cdots +  \alpha_{g_n}g_n   \in RG.$  Define the matrix $\sigma(v) \in M_n(R)$ to be
\begin{equation} \label{equation:construction}
\sigma(v) = 
\left( 
\begin{array}{ccccc}
\alpha_{g_1^{-1} g_1} & \alpha_{g_1^{-1} g_2} & \alpha_{g_1^{-1} g_3} & \dots & \alpha_{g_1^{-1} g_n}  \\
\alpha_{g_2^{-1} g_1} & \alpha_{g_2^{-1} g_2} & \alpha_{g_2^{-1} g_3} & \dots & \alpha_{g_2^{-1} g_n}  \\
\vdots  & \vdots & \vdots & \vdots & \vdots \\ 
\alpha_{g_n^{-1} g_1} & \alpha_{g_n^{-1} g_2} & \alpha_{g_n^{-1} g_3} & \dots & \alpha_{g_n^{-1} g_n}  \end{array}
\right).
\end{equation}

We note that the elements $g_1^{-1}, g_2^{-1}, \dots, g_n^{-1}$ are simply the elements of the group $G$ in some order.

Define the following code over the ring $R$ for a given element $v \in RG$. 
\begin{equation}
C(v) = \langle \sigma (v) \rangle.
\end{equation}
Specifically,  the code is formed by taking the row space of $\sigma(v)$ over the ring $R$.
Given that it is the span of a set of vectors, the code is necessarily linear.  Given the structure of the matrix it follows immediately that the code is held invariant by the action of the group $G$.  
Therefore, we have that $C(v)$ is a $G$-code, namely a code that corresponds to an ideal in $RG$.

We do not claim that the rows of the matrix $\sigma(v)$ are necessarily linearly independent, in general, they will not be.  If they are then the generated code is the ambient space. 

\begin{example} As a very simple example, let $G$ be any finite group of odd order, $R$ be a finite Frobenius ring with characterisitc $2$,  and let $v= \sum _{g \in G} v_g \in RG$. Then $C= \langle \sigma(v) \rangle $ is the code generated by the all-one vector.  Since the length is odd, this code has a trivial intersection with its orthogonal, giving that $C$ is an LCD $G$-code.
\end{example}

For an element $v = \sum \alpha_i g_i \in RG$, define the element $v^T \in RG$ as 
$v^T = \sum \alpha_i g_i^{-1}.$   This is sometimes known as the canonical involution for the group ring.

In \cite{Dougherty5}, the following is proven. 
Let $R$ be a finite commutative Frobenius ring and let $G$ be a group of order $n$. 
For an element  $v \in RG$, we have that $\sigma(v)^T = \sigma(v^T).$

From \cite{Dougherty5}, we have the following.
Let $R$ be a finite commutative Frobenius ring and let $G$ be a group of order $n$. 
Then the map $\sigma: RG \rightarrow M_n(R)$ is an injective ring homomorphism.  

Given that we have the map $\sigma: RG \rightarrow M_n(R)$ is an injective ring homomorphism one may be tempted to find elements in $RG$ that satisfy $vv^T =1$, so that $\sigma(v) \sigma(v^T) = \sigma(v) \sigma(v)^T = I_n.$  However, in this case, given the structure of $\sigma(v)$ this will only give trivial LCD codes.   Therefore, the approach is that we need to find $v$ so that $\sigma(v)$ does not have a full rank but still produces an LCD code.

\section{Construction of group LCD codes} 

In this section, we show what one needs in terms of the group ring element $v,$ so that $\sigma(v)$ does not have a full rank but still produces an LCD code.

Let $Hull(v) = C(v) \cap C(v)^\perp.$  

\begin{lem}
Let $v\in RG$, then $Hull(v)$ is a $G$-code.
\end{lem} 
\begin{proof}
We know from \cite{us}, that $C(v)^\perp$ is a $G$-code.  Then  if $\vv \in C(v)$, then $\sigma \vv \in C(v)$ for all $\sigma \in G$ and if $\vv \in C(v)^\perp$, then $\sigma \vv \in C(v)^\perp$ for all $\sigma \in G$.
This gives that if $\vv \in Hull(v)$, then $\sigma \vv \in Hull(v)$ for all $\sigma \in G$
\end{proof} 

For an LCD code we have that $Hull(v) = \{ \bf 0 \}.$

Let $I$ be an ideal in a group ring $RG$.  Define ${\cal R}(I) = \{ w \ | \ vw =0, \forall v \in I \}.$  It follows immediately that ${\cal R}(I)$ is an ideal of $RG.$

Let $v = a_{g_1} g_1 +a_{g_2} g_2 + \dots +a_{g_n} g_n \in RG$ and $C(v)$.
Let $\Psi:RG \rightarrow R^n$ be the canonical map defined by 
$$ \Psi(a_{g_1} g_1 +a_{g_2} g_2 + \dots +a_{g_n} g_n) = (a_{g_1},a_{g_2}, \dots, a_{g_n}).$$
It is shown in \cite{us}, that  $\Psi^{-1}(C(v))$ is an ideal of $RG.$

 \begin{lem} \label{lemmaone} 
 Let $\vw = \Psi(w).$   Then 
 ${ w^T} \in {\cal R}(I(v))$ if and only if $\vw \in C(v)^\perp.$
 \end{lem}
 \begin{proof}  
  Let $I(v)$ be the ideal $\Psi^{-1}(C(v))$.  
Let $\vw  = (w_1,w_2,\dots,w_n) \in C^\perp.$  Then 
\begin{equation}
[  (a_{g_j^{-1} g_1} , a_{g_j^{-1} g_2}, \dots,  a_{g_j^{-1} g_n} ),( w_1,w_2,\dots,w_n)] = 0, \forall j.  
\end{equation} 
This gives that 
\begin{equation}
\sum_{i=1}^n   a_{g_j^{-1} g_i}w_i  = 0, \ \forall j.
\end{equation}

Then
\begin{equation}
\sum_{i=1}^n   a_{g_j^{-1} g_i} w_i = 0 \implies \sum_{i=1}^n a_{ g_j^{-1} g_i} b_{ g_i^{-1} } =0.
\end{equation} 
Then $g_j^{-1} g_i g_i^{-1} = g_j^{-1}$, hence this is the coefficient of $g_j^{-1}$ in the product of $ w^T$ and 
$g_j^{-1} v$.  This gives that $ w^T \in {\cal R}(I(v))$ if and only if $\vw \in C(v)^\perp.$   
\end{proof}

Let ${\cal R} (I(v))^T = \{ w^t \ | \ w \in {\cal R} (I(v))$ and  let  ${\cal H}(v) = I(v) \cap  {\cal R}({\cal H}(v))^T.$

 \begin{lem}\label{lemmatwo}
 Let $v\in RG$, then $\Psi( {\cal R} (I(v))^T )= C^\perp$ and $\Psi({\cal H}(v) ) = Hull(v)$. \end{lem} 
 \begin{proof}
 Follows from a direct application of Lemma~\ref{lemmaone}.
 \end{proof} 
 
 We have the following diagrams.

 
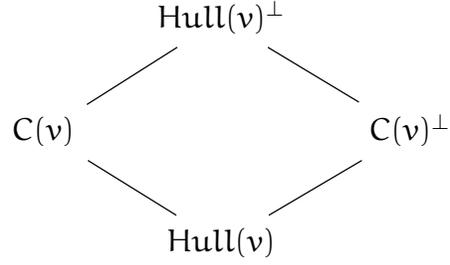
\begin{center} 
 \begin{tikzcd}[column sep=scriptsize]
 & Hull(v)^\perp \arrow[dr, no head ] \arrow[dl, no head ] \\
C(v)  \arrow[dr, no head ] {}
& & C(v)^\perp  \arrow[dl, no head] \\
& Hull(v)  \\
\end{tikzcd}

  
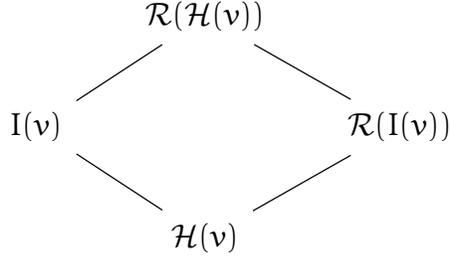
\begin{tikzcd}[column sep=scriptsize]
 & {\cal R}({\cal H}(v)) \arrow[dr, no head ] \arrow[dl, no head ] \\
I(v)  \arrow[dr, no head ] {}
& & {\cal R}(I(v))  \arrow[dl, no head] \\
&  {\cal H}(v)    \\
\end{tikzcd}
\end{center}

We now state our main result.

\begin{theorem}\label{MainThm}
Let $v \in RG.$  If $C(v)$ is a non-trivial  LCD code then $v= v^T.$  
\end{theorem} 
\begin{proof}
Let $R(I(v))^T= \langle w_1^T,w_2^T,\dots,w_s^T \rangle.$  Assume 
$C(v) $ is an LCD code.  Then
$ \langle w_1^T,w_2^T,\dots,w_s^T \rangle = RG.$  Hence no element divides $v$ and each $w_i^T.$  
If $w_i^T \in {\cal R}(I)^T$ then $w_i \in {\cal R}(I)$ which gives that $v w_i =0$ for all $i$.  It follows that $v^Tw_i^T =0$ for all $i$.

If $vw_i^T = \alpha \neq 0$ then $\alpha \in I(v)$ and $\alpha \in {\cal R}(I(v))^T$ which gives $\alpha \in {cal H}(v)$ which is a contradiction since this ideal is trivial by Lemma~\ref{lemmatwo}.  This gives that $v w_i^T =0$ for all $i$.  

Then we have that $vw_i^T = v^T w_i^T$ which gives
$v(w_1^T+w_2^T + \dots + w_s^T) = v^T (w_1^T+w_2^T + \dots + w_s^T)$ and $v$ and $(w_1^T+w_2^T + \dots + w_s^T) $ can have no common factor.  This gives that $v=v^T.$  
\end{proof}

This result allows us to reduce the search for LCD codes significantly.  Namely, we need only search through $v$ with $v=v^T.$  

It follows then that if $C(v)$ is an LCD code then $C(v) = C(v^T).$ 

We note that the theorem is not a biconditional.  Consider the vector $v=1+1g_1 + 1 g_2 + 1 g_3 \in \FF_2C_4$.
Then $v=v^T$ but $Hull(v) =C(v)$ and so $C(v)$ is not an LCD code.  

We now present an example.

\begin{example}
Let $v=\alpha_{1}1+\alpha_{a}a+\alpha_{a^2}+\alpha_{b}b+\alpha_{ab}ab+\alpha_{a^2b}a^2b \in \mathbb{F}_2D_6,$ where $D_6=\{1,a,a^a,b,ab,a^2b\}$ is the dihedral group with $6$ elements. Then
$$\sigma(v)=\begin{pmatrix}
\alpha_1&\alpha_a&\alpha_{a^2}&\alpha_{b}&\alpha_{ab}&\alpha_{a^2b}\\
\alpha_{a^2}&\alpha_1&\alpha_a&\alpha_{a^2b}&\alpha_b&\alpha_{ab}\\
\alpha_{a}&\alpha_{a^2}&\alpha_1&\alpha_{ab}&\alpha_{a^2b}&\alpha_b\\
\alpha_b&\alpha_{a^2b}&\alpha_{ab}&\alpha_1&\alpha_{a^2}&\alpha_a\\
\alpha_{ab}&\alpha_b&\alpha_{a^2b}&\alpha_{a}&\alpha_1&\alpha_{a^2}\\
\alpha_{a^2b}&\alpha_{ab}&\alpha_b&\alpha_{a^2}&\alpha_a&\alpha_1
\end{pmatrix}.$$
Now, if we add the condition $v=v^T$ in, that is, we want the coefficient of the group element $g_i$ to be the same as the coefficient of the group element $g_i^{-1}$, the above matrix becomes:
$$\sigma'(v)=\begin{pmatrix}
\alpha_1&\alpha_a&\alpha_{a}&\alpha_{b}&\alpha_{ab}&\alpha_{a^2b}\\
\alpha_{a}&\alpha_1&\alpha_a&\alpha_{a^2b}&\alpha_b&\alpha_{ab}\\
\alpha_{a}&\alpha_{a}&\alpha_1&\alpha_{ab}&\alpha_{a^2b}&\alpha_b\\
\alpha_b&\alpha_{a^2b}&\alpha_{ab}&\alpha_1&\alpha_{a}&\alpha_a\\
\alpha_{ab}&\alpha_b&\alpha_{a^2b}&\alpha_{a}&\alpha_1&\alpha_{a}\\
\alpha_{a^2b}&\alpha_{ab}&\alpha_b&\alpha_{a}&\alpha_a&\alpha_1
\end{pmatrix},$$
since $1^{-1}=1, a^{-1}=a^2, (a^2)^{-1}=a, b^{-1}=b, (ab)^{-1}=ab$ and $(a^2b)^{-1}=a^2b.$

One can see that the search field in the matrix $\sigma(v)$ is $2^{6}$ and $2^5$ in the matrix $\sigma'(v).$ We now search for dihedral LCD group codes of length 6 using both matrices; $\sigma(v)$ and $\sigma'(v)$ and we obtain the following:
$$C_1=\begin{pmatrix}
0&0&0&0&0&0\\
0&0&0&0&0&0\\
0&0&0&0&0&0\\
0&0&0&0&0&0\\
0&0&0&0&0&0\\
0&0&0&0&0&0
\end{pmatrix}, C_2=\begin{pmatrix}
0&0&0&0&0&1\\
0&0&0&1&0&0\\
0&0&0&0&1&0\\
0&1&0&0&0&0\\
0&0&1&0&0&0\\
1&0&0&0&0&0
\end{pmatrix}, C_3=\begin{pmatrix}
0&0&0&0&1&1\\
0&0&0&1&0&1\\
0&0&0&1&1&0\\
0&1&1&0&0&0\\
1&0&1&0&0&0\\
1&1&0&0&0&0
\end{pmatrix}$$
and
$$C_4=\begin{pmatrix}
0&0&0&1&1&1\\
0&0&0&1&1&1\\
0&0&0&1&1&1\\
1&1&1&0&0&0\\
1&1&1&0&0&0\\
1&1&1&0&0&0
\end{pmatrix}.$$
The same 4 dihedral LCD codes were obtained from both matrices; $\sigma(v)$ and $\sigma'(v).$ Of course, there were fewer calculations needed with the second matrix $\sigma'(v)$ which highlights the importance of Theorem~\ref{MainThm}.
\end{example}

\section{Construction of $\sigma(v)$ with $v=v^T$ for various groups}\label{v=v^T}

In this section, we want to generalize the forms of the matrices $\sigma(v)$ with the condition $v=v^T$ for various groups. That is, in this section, we present a number of matrices $\sigma(v),$ where $v=\alpha_{g_1}g_1+\alpha_{g_2}g_2+\dots+\alpha_{g_n}g_n \in RG$ with the coefficient of $g_i$ being the same as the coefficient of $g_i^{-1}.$ This will enable us to calculate group LCD codes of different lengths for different groups in less and more practical time.
\begin{enumerate}
\item[1.] Let $C_n=\langle a \ | \ a^n=e \rangle$ be the cyclic group of order $n$ with $n=2k,$ for $k \in \mathbb{Z}^+.$ Let $v=\sum_{i=1}^n \alpha_{a^{i-1}} a^{i-1} \in RC_n.$ Then 

\begin{equation}\label{cyclic1}
\sigma(v)=circ(\alpha_{e},\alpha_{a},\alpha_{a^2},\dots,\alpha_{a^{\frac{n}{2}}},\alpha_{a^{\frac{n}{2}-1}},\alpha_{a^{\frac{n}{2}-2}},\dots,\alpha_{a^2},\alpha_{a}).
\end{equation}

\item[2.] Let $C_n= \langle a \ | \ a^n=e \rangle$ be the cyclic group of order $n$ with $n=2k+1,$ for $k \in \mathbb{Z}^+.$ Let $v=\sum_{i=1}^n \alpha_{a^{i-1}} a^{i-1} \in RC_n.$ Then

\begin{equation}\label{cyclic2}
\sigma(v)=circ(\alpha_{e},\alpha_{a},\alpha_{a^2},\dots,\alpha_{a^{\frac{n-1}{2}}},\alpha_{a^{\frac{n-1}{2}}},\alpha_{a^{\frac{n-1}{2}-1}},\alpha_{a^{\frac{n-1}{2}-2}},\dots,\alpha_{a^2},\alpha_{a}).
\end{equation}

\item[3.] Let $D_{2n}=\langle a,b \ | \ a^n=b^2=e, a^b=a^{-1}\rangle$ be the dihedral group of order $2n$ with $n=2k,$ for $k \in \mathbb{Z}^+.$ Let $v=\sum_{i=0}^n \sum_{j=0}^1 \alpha_{a^ib^j}a^ib^j \in RD_{2n}.$ Then

\begin{equation}\label{dihedral1a}
\sigma(v)=\begin{pmatrix}
A&B\\
B^T&A^T
\end{pmatrix},
\end{equation}
where 
$$A=circ(\alpha_{e},\alpha_{a},\alpha_{a^2},\dots,\alpha_{a^{\frac{n}{2}}},\alpha_{a^{\frac{n}{2}-1}},\alpha_{a^{\frac{n}{2}-2}},\dots,\alpha_{a^2},\alpha_a),$$
$$B=circ(\alpha_b,\alpha_{ab},\alpha_{a^2b},\dots,\alpha_{a^{n-1}b}).$$

\item[4.] Let $D_{2n}=\langle a,b \ | \ a^n=b^2=e, a^b=a^{-1} \rangle$ be the dihedral group of order $2n$ with $n=2k+1,$ for $k \in \mathbb{Z}^+.$ Let $v=\sum_{i=0}^n \sum_{j=0}^1 \alpha_{a^ib^j}a^ib^j \in RD_{2n}.$ Then

\begin{equation}\label{dihedral1b}
\sigma(v)=\begin{pmatrix}
A&B\\
B^T&A^T
\end{pmatrix},
\end{equation}
where 
$$A=circ(\alpha_{e},\alpha_{a},\alpha_{a^2},\dots,\alpha_{a^{\frac{n-1}{2}}},\alpha_{a^{\frac{n-1}{2}}},\alpha_{a^{\frac{n-1}{2}-1}},\alpha_{a^{\frac{n-1}{2}-2}},\dots,\alpha_{a^2},\alpha_{a}),$$
$$B=circ(\alpha_b,\alpha_{ab},\alpha_{a^2b},\dots,\alpha_{a^{n-1}b}).$$

\item[5.] Let $D_{2n}=\langle a,b \ | \ a^n=b^2=e, a^b=a^{-1}\rangle$ be the dihedral group of order $2n$ with $n=2k,$ for $k \in \mathbb{Z}^+.$ Let $v=\sum_{i=0}^n \sum_{j=0}^1 \alpha_{b^ja^i}b^ja^i \in RD_{2n}.$ Then

\begin{equation}\label{dihedral2a}
\sigma(v)=\begin{pmatrix}
A&B\\
B&A
\end{pmatrix},
\end{equation}
where 
$$A=circ(\alpha_{e},\alpha_{a},\alpha_{a^2},\dots,\alpha_{a^{\frac{n}{2}}},\alpha_{a^{\frac{n}{2}-1}},\alpha_{a^{\frac{n}{2}-2}},\dots,\alpha_{a^2},\alpha_a),$$
$$B=revcirc(\alpha_b,\alpha_{ba},\alpha_{ba^2},\dots,\alpha_{ba^{n-1}}).$$

\item[6.] Let $D_{2n}=\langle a,b \ | \ a^n=b^2=e, a^b=a^{-1}\rangle$ be the dihedral group of order $2n$ with $n=2k+1,$ for $k \in \mathbb{Z}^+.$ Let $v=\sum_{i=0}^n \sum_{j=0}^1 \alpha_{b^ja^i}b^ja^i \in RD_{2n}.$ Then

\begin{equation}\label{dihedral2b}
\sigma(v)=\begin{pmatrix}
A&B\\
B&A
\end{pmatrix},
\end{equation}
where 
$$A=circ(\alpha_{e},\alpha_{a},\alpha_{a^2},\dots,\alpha_{a^{\frac{n-1}{2}}},\alpha_{a^{\frac{n-1}{2}}},\alpha_{a^{\frac{n-1}{2}-1}},\alpha_{a^{\frac{n-1}{2}-2}},\dots,\alpha_{a^2},\alpha_{a}),$$
$$B=revcirc(\alpha_b,\alpha_{ba},\alpha_{ba^2},\dots,\alpha_{ba^{n-1}}).$$

\end{enumerate}

\section{Construction of group reversible LCD codes}

In this section, we show that for certain groups and with a fixed listing of their elements in the group ring $RG,$ one can construct group reversible LCD codes with our method. We start with a definition from \cite{ReversibleDNACodes}.

\begin{defi}
A code $C$ is said to be reversible of index $\alpha$ if $\mathbf{a}_i$ is a vector of length $\alpha$ and $\vc^{\alpha}=(\mathbf{a}_0,\mathbf{a}_1,\dots,\mathbf{a}_{s-1}) \in C$ implies that $(\vc^{\alpha})^r=(\mathbf{a}_{s-1},\mathbf{a}_{s-2},\dots,\mathbf{a}_{1},\mathbf{a}_{0}) \in C.$
\end{defi}

Also in \cite{ReversibleDNACodes}, the following is shown.

\begin{theorem}\label{ReversibleGcode}
Let $R$ be a finite ring. Let $G$ be a finite group of order $n=2\ell$ and let $H=\{e,h_1,h_2,\dots ,h_{\ell-1}\}$ be a subgroup of index 2 in $G.$ Let $\beta \notin H$ be an element in $G$ with $\beta^{-1}=\beta.$ List the elements of $G$ as
$$\{ e, h_1, \dots, h_{\ell-1}, \beta h_{\ell-1}, \beta h_{\ell-2}, \beta h_2, \beta h_1,\beta\},$$
then any linear $G$-code in $R^n$ (a left ideal in $RG$) is a reversible code of index 1.
\end{theorem}

Combining the above result with Theorem~\ref{MainThm}, we get the following

\begin{theorem}\label{RevLCDThm1}
Let $R$ be a finite ring. Let $G$ be a finite group of order $n=2\ell$ and let $H=\{e,h_1,h_2,\dots ,h_{\ell-1}\}$ be a subgroup of index 2 in $G.$ Let $\beta \notin H$ be an element in $G$ with $\beta^{-1}=\beta.$ Let $v \in RG$ with $v=v^T$ such that the elements of $G$ in $v$ are listed as
$$\{ e, h_1, \dots, h_{\ell-1}, \beta h_{\ell-1}, \beta h_{\ell-2}, \beta h_2, \beta h_1,\beta\}.$$
Then $C(v)$ is a non-trivial reversible LCD code of index 1.
\begin{proof}
The reversibility follows from Theorem~\ref{MainThm} and the second part follows from Theorem~\ref{ReversibleGcode}.
\end{proof}
\end{theorem}

The above result shows that with a careful selection of the element $v$ in the group ring $RG,$ one can construct a group code that is both, reversible and LCD.

\begin{theorem}
Let $G$ be a finite group of order $n=2\ell$ and let $H=\{e,h_1,h_2,\dots ,h_{\ell-1}\}$ be a subgroup of index 2 in $G.$ Let $\beta \notin H$ be an element in $G$ with $\beta^{-1}=\beta.$ Let $v \in RG$ with $v=v^T$ such that the elements of $G$ in $v$ are listed as
$$\{ e, h_1, \dots, h_{\ell-1}, \beta h_{\ell-1}, \beta h_{\ell-2}, \beta h_2, \beta h_1,\beta\}.$$
If $C(v)$ is a linear $G$-code in $R_k^n$ (a left ideal in $R_kG$), then $\phi(C(v))$ is a non-trivial reversible LCD code of index $2^k.$
\begin{proof}
By Theorem~\ref{RevLCDThm1}, we have that $C$ is a non-trivial reversible LCD code of index 1. Therefore, if $(c_0,c_1,\dots,c_{n-1}) \in C(v)$ we have that $(c_{n-1},c_{n-2},\dots,c_1,c_0) \in C(v),$ where $c_i \in R_k.$ Then $\phi(c_i)$ is a vector of length $2^k.$ This gives that
$$(\phi(c_0),\phi(c_1),\dots,\phi(c_{n-1})) \in \phi(C)$$
and then
$$(\phi(c_{n-1}),\phi(c_{n-2}),\dots,\phi(c_1),\phi(c_0)) \in \phi(C).$$
This gives the result.
\end{proof}
\end{theorem}

\section{Computational Results}

In this section, we construct many group LCD and group reversible LCD codes using the code construction $\langle \sigma(v) \rangle,$ given in Equation~(\ref{equation:construction}), where $\sigma(v)$ are some of the $n \times n$ matrices described in Section~\ref{v=v^T}. The searches are performed in the software package MAGMA (\cite{MAGMA}). We only tabulate the lengths, dimensions and the largest minimum distance of the codes that we construct. Codes with the largest minimum distance that are optimal (according to \cite{Bouyuklieva}) are written in bold. The generator matrices of these codes and their corresponding weight enumerators can be found at \cite{GenMatr.}. 

\subsection{Group LCD Codes}

In this section we search for group LCD codes using the code construction $\langle \sigma(v) \rangle,$ given in Equation~(\ref{equation:construction}), where $\sigma(v)$ are some of the $n \times n$ matrices described in Section~\ref{v=v^T}.

\begin{enumerate}
\item[1.] Let $v \in \mathbb{F}_2C_{n}$ and let $\mathcal{G}_1=\langle \sigma(v) \rangle,$ where $\sigma(v)$ is the matrix given in Equation~(\ref{cyclic1}). We set $n$ to be $28, 30, 34, 36$ and $38.$ We list the results in Table~\ref{Table1}, Table~\ref{Table2}, Table~\ref{Table3}, Table~\ref{Table4} and Table~\ref{Table5} respectively.

\begin{table}[h!]
\caption{Group LCD Codes from $\mathcal{G}_1 \ (n=28)$  over $\mathbb{F}_2$}
\resizebox{1\textwidth}{!}{\begin{minipage}{\textwidth}\label{Table1}
\centering
\begin{tabular}{|c|c|c|c|}
\hline
$\mathcal{C}_i$ & $[n,k,d]$   & $\mathcal{C}_i$ & $[n,k,d]$   \\ \hline
$\mathcal{C}_1$ & $[28,4,7]$ & $\mathcal{C}_2$ & $[28,24,2]$ \\ \hline
\end{tabular}
\end{minipage}}
\end{table}

\begin{table}[h!]
\caption{Group LCD Codes from $\mathcal{G}_1 \ (n=30)$  over $\mathbb{F}_2$}
\resizebox{1\textwidth}{!}{\begin{minipage}{\textwidth}\label{Table2}
\centering
\begin{tabular}{|c|c|c|c|}
\hline
$\mathcal{C}_i$ & $[n,k,d]$   & $\mathcal{C}_i$ & $[n,k,d]$   \\ \hline
$\mathcal{C}_1$ & $[30,2,15]$ & $\mathcal{C}_8$ & $[30,16,4]$ \\ \hline
$\mathcal{C}_2$ & $[30,4,10]$ & $\mathcal{C}_{9}$ & $[30,18,4]$ \\ \hline
$\mathcal{C}_3$ & $[30,6,5]$  & $\mathcal{C}_{10}$ & $[30,20,2]$ \\ \hline
$\mathcal{C}_4$ & $[30,8,6]$  & $\mathcal{C}_{11}$ & $[30,22,2]$ \\ \hline
$\mathcal{C}_5$ & $[30,10,3]$ & $\mathcal{C}_{12}$ & $[30,24,2]$ \\ \hline
$\mathcal{C}_6$ & $[30,12,6]$ & $\mathcal{C}_{13}$ & $[30,26,2]$ \\ \hline
$\mathcal{C}_7$ & $[30,14,3]$  & $\mathcal{C}_{14}$ & $[30,28,2]$ \\ \hline
\end{tabular}
\end{minipage}}
\end{table}

\begin{table}[h!]
\caption{Group LCD Codes from $\mathcal{G}_1 \ (n=34)$  over $\mathbb{F}_2$}
\resizebox{1\textwidth}{!}{\begin{minipage}{\textwidth}\label{Table3}
\centering
\begin{tabular}{|c|c|c|c|}
\hline
$\mathcal{C}_i$ & $[n,k,d]$   & $\mathcal{C}_i$ & $[n,k,d]$   \\ \hline
$\mathcal{C}_1$ & $[34,2,17]$ & $\mathcal{C}_3$ & $[34,18,5]$ \\ \hline
$\mathcal{C}_2$ & $[34,16,6]$ & $\mathcal{C}_{4}$ & $[34,32,2]$ \\ \hline
\end{tabular}
\end{minipage}}
\end{table}

\begin{table}[h!]
\caption{Group LCD Codes from $\mathcal{G}_1 \ (n=36)$  over $\mathbb{F}_2$}
\resizebox{1\textwidth}{!}{\begin{minipage}{\textwidth}\label{Table4}
\centering
\begin{tabular}{|c|c|c|c|}
\hline
$\mathcal{C}_i$ & $[n,k,d]$   & $\mathcal{C}_i$ & $[n,k,d]$   \\ \hline
$\mathcal{C}_1$ & $[36,4,9]$ & $\mathcal{C}_4$ & $[36,24,2]$ \\ \hline
$\mathcal{C}_2$ & $[36,8,6]$ & $\mathcal{C}_{5}$ & $[36,28,2]$ \\ \hline
$\mathcal{C}_3$ & $[36,12,3]$  & $\mathcal{C}_{6}$ & $[36,32,2]$ \\ \hline
\end{tabular}
\end{minipage}}
\end{table}
\newpage
\begin{table}[h!]
\caption{Group LCD Codes from $\mathcal{G}_1 \ (n=38)$  over $\mathbb{F}_2$}
\resizebox{1\textwidth}{!}{\begin{minipage}{\textwidth}\label{Table5}
\centering
\begin{tabular}{|c|c|c|c|}
\hline
$\mathcal{C}_i$ & $[n,k,d]$   & $\mathcal{C}_i$ & $[n,k,d]$   \\ \hline
$\mathcal{C}_1$ & $[38,2,19]$ & $\mathcal{C}_2$ & $[38,36,2]$ \\ \hline

\end{tabular}
\end{minipage}}
\end{table}

\item[2.] Let $v \in \mathbb{F}_2C_{n}$ and let $\mathcal{G}_2=\langle \sigma(v) \rangle,$ where $\sigma(v)$ is the matrix given in Equation~(\ref{cyclic2}). We set $n$ to be $29, 31, 33, 35, 37$ and $39.$ We list the results in Table~\ref{Table6}, Table~\ref{Table7}, Table~\ref{Table8}, Table~\ref{Table9}, Table~\ref{Table10} and Table~\ref{Table11} respectively.

\begin{table}[h!]
\caption{Group LCD Codes from $\mathcal{G}_2 \ (n=29)$  over $\mathbb{F}_2$}
\resizebox{1\textwidth}{!}{\begin{minipage}{\textwidth}\label{Table6}
\centering
\begin{tabular}{|c|c|}
\hline
$\mathcal{C}_i$ & $[n,k,d]$   \\ \hline
$\mathcal{C}_1$ & $[29,28,2]$  \\ \hline
\end{tabular}
\end{minipage}}
\end{table}

\begin{table}[h!]
\caption{Group LCD Codes from $\mathcal{G}_2 \ (n=31)$  over $\mathbb{F}_2$}
\resizebox{1\textwidth}{!}{\begin{minipage}{\textwidth}\label{Table7}
\centering
\begin{tabular}{|c|c|c|c|}
\hline
$\mathcal{C}_i$ & $[n,k,d]$   & $\mathcal{C}_i$ & $[n,k,d]$   \\ \hline
$\mathcal{C}_1$ & $\mathbf{[31,10,10]}$ & $\mathcal{C}_4$ & $\mathbf{[31,21,5]}$ \\ \hline
$\mathcal{C}_2$ & $\mathbf{[31,11,10]}$ & $\mathcal{C}_5$ & $[31,30,2]$ \\ \hline
$\mathcal{C}_3$ & $\mathbf{[31,20,6]}$ &  &  \\ \hline
\end{tabular}
\end{minipage}}
\end{table}

\begin{table}[h!]
\caption{Group LCD Codes from $\mathcal{G}_2 \ (n=33)$  over $\mathbb{F}_2$}
\resizebox{1\textwidth}{!}{\begin{minipage}{\textwidth}\label{Table8}
\centering
\begin{tabular}{|c|c|c|c|}
\hline
$\mathcal{C}_i$ & $[n,k,d]$   & $\mathcal{C}_i$ & $[n,k,d]$   \\ \hline
$\mathcal{C}_1$ & $[33,2,22]$ & $\mathcal{C}_{12}$ & $\mathbf{[33,20,6]}$ \\ \hline
$\mathcal{C}_2$ & $[33,3,11]$ & $\mathcal{C}_{13}$ & $[33,21,3]$ \\ \hline
$\mathcal{C}_3$ & $[33,10,6]$ & $\mathcal{C}_{14}$ & $[33,21,4]$ \\ \hline
$\mathcal{C}_4$ & $\mathbf{[33,10,12]}$ & $\mathcal{C}_{15}$ & $[33,22,2]$ \\ \hline
$\mathcal{C}_5$ & $[33,11,3]$ & $\mathcal{C}_{16}$ & $\mathbf{[33,22,6]}$ \\ \hline
$\mathcal{C}_6$ & $\mathbf{[33,11,11]}$ & $\mathcal{C}_{17}$ & $[33,23,2]$ \\ \hline
$\mathcal{C}_7$ & $[33,12,6]$ & $\mathcal{C}_{18}$ & $[33,23,3]$ \\ \hline
$\mathcal{C}_8$ & $\mathbf{[33,12,10]}$ & $\mathcal{C}_{19}$ & $[33,30,2]$ \\ \hline
$\mathcal{C}_9$ & $[33,13,3]$ & $\mathcal{C}_{20}$ & $[33,31,2]$ \\ \hline
$\mathcal{C}_{10}$ & $\mathbf{[33,13,10]}$ & $\mathcal{C}_{21}$ & $[33,32,2]$ \\ \hline
$\mathcal{C}_{11}$ & $[33,20,4]$ &  &  \\ \hline
\end{tabular}
\end{minipage}}
\end{table}
\newpage

\begin{table}[h!]
\caption{Group LCD Codes from $\mathcal{G}_2 \ (n=35)$  over $\mathbb{F}_2$}
\resizebox{1\textwidth}{!}{\begin{minipage}{\textwidth}\label{Table9}
\centering
\begin{tabular}{|c|c|c|c|}
\hline
$\mathcal{C}_i$ & $[n,k,d]$   & $\mathcal{C}_i$ & $[n,k,d]$   \\ \hline
$\mathcal{C}_1$ & $[35,4,14]$ & $\mathcal{C}_{8}$ & $\mathbf{[35,25,4]}$ \\ \hline
$\mathcal{C}_2$ & $[35,5,7]$ & $\mathcal{C}_{9}$ & $[35,28,2]$ \\ \hline
$\mathcal{C}_3$ & $[35,6,10]$ & $\mathcal{C}_{10}$ & $[35,29,2]$ \\ \hline
$\mathcal{C}_4$ & $[35,7,5]$ & $\mathcal{C}_{11}$ & $[35,30,2]$ \\ \hline
$\mathcal{C}_5$ & $[35,10,10]$ & $\mathcal{C}_{12}$ & $[35,31,2]$ \\ \hline
$\mathcal{C}_6$ & $[35,11,5]$ & $\mathcal{C}_{13}$ & $[35,34,2]$ \\ \hline
$\mathcal{C}_{7}$ & $[35,24,4]$ &  &  \\ \hline
\end{tabular}
\end{minipage}}
\end{table}

\begin{table}[h!]
\caption{Group LCD Codes from $\mathcal{G}_2 \ (n=37)$  over $\mathbb{F}_2$}
\resizebox{1\textwidth}{!}{\begin{minipage}{\textwidth}\label{Table10}
\centering
\begin{tabular}{|c|c|}
\hline
$\mathcal{C}_i$ & $[n,k,d]$   \\ \hline
$\mathcal{C}_1$ & $[37,36,2]$  \\ \hline
\end{tabular}
\end{minipage}}
\end{table}

\begin{table}[h!]
\caption{Group LCD Codes from $\mathcal{G}_2 \ (n=39)$  over $\mathbb{F}_2$}
\resizebox{1\textwidth}{!}{\begin{minipage}{\textwidth}\label{Table11}
\centering
\begin{tabular}{|c|c|c|c|}
\hline
$\mathcal{C}_i$ & $[n,k,d]$   & $\mathcal{C}_i$ & $[n,k,d]$   \\ \hline
$\mathcal{C}_1$ & $[39,2,26]$ & $\mathcal{C}_{8}$ & $[39,25,4]$ \\ \hline
$\mathcal{C}_2$ & $[39,3,13]$ & $\mathcal{C}_{9}$ & $[39,26,2]$ \\ \hline
$\mathcal{C}_3$ & $[39,12,6]$ & $\mathcal{C}_{10}$ & $[39,27,2]$ \\ \hline
$\mathcal{C}_4$ & $[39,13,3]$ & $\mathcal{C}_{11}$ & $[39,36,2]$ \\ \hline
$\mathcal{C}_5$ & $[39,14,6]$ & $\mathcal{C}_{12}$ & $[39,37,2]$ \\ \hline
$\mathcal{C}_6$ & $[39,15,3]$ & $\mathcal{C}_{13}$ & $[39,38,2]$ \\ \hline
$\mathcal{C}_{7}$ & $[39,24,4]$ &  &  \\ \hline
\end{tabular}
\end{minipage}}
\end{table}

\item[3.] Let $v \in \mathbb{F}_2D_{24}$ and let $\mathcal{G}_3=\langle \sigma(v) \rangle,$ where $\sigma(v)$ is the matrix given in Equation~(\ref{dihedral2a}). We list our results in Table~\ref{Table12}.

\begin{table}[h!]
\caption{Group LCD Codes from $\mathcal{G}_3$ over $\mathbb{F}_2$}
\resizebox{1\textwidth}{!}{\begin{minipage}{\textwidth}\label{Table12}
\centering
\begin{tabular}{|c|c|c|c|}
\hline
$\mathcal{C}_i$ & $[n,k,d]$   & $\mathcal{C}_i$ & $[n,k,d]$   \\ \hline
$\mathcal{C}_1$ & $[24,8,3]$ & $\mathcal{C}_2$ & $[24,16,2]$ \\ \hline
\end{tabular}
\end{minipage}}
\end{table}

\item[4.] Let $v \in \mathbb{F}_2D_{30}$ and let $\mathcal{G}_4=\langle \sigma(v) \rangle,$ where $\sigma(v)$ is the matrix given in Equation~(\ref{dihedral2b}). We list our results in Table~\ref{Table13}.
\newpage
\begin{table}[h!]
\caption{Group Reversible LCD Codes from $\mathcal{G}_4$ over $\mathbb{F}_2$}
\resizebox{1\textwidth}{!}{\begin{minipage}{\textwidth}\label{Table13}
\centering
\begin{tabular}{|c|c|c|c|}
\hline
$\mathcal{C}_i$ & $[n,k,d]$   & $\mathcal{C}_i$ & $[n,k,d]$   \\ \hline
$\mathcal{C}_1$ & $[30,2,15]$ & $\mathcal{C}_8$ & $[30,16,4]$ \\ \hline
$\mathcal{C}_2$ & $[30,4,10]$ & $\mathcal{C}_{9}$ & $[30,18,4]$ \\ \hline
$\mathcal{C}_3$ & $[30,6,5]$  & $\mathcal{C}_{10}$ & $[30,20,2]$ \\ \hline
$\mathcal{C}_4$ & $[30,8,6]$  & $\mathcal{C}_{11}$ & $[30,22,2]$ \\ \hline
$\mathcal{C}_5$ & $[30,10,3]$ & $\mathcal{C}_{12}$ & $[30,24,2]$ \\ \hline
$\mathcal{C}_6$ & $[30,12,6]$ & $\mathcal{C}_{13}$ & $[30,26,2]$ \\ \hline
$\mathcal{C}_7$ & $[30,14,3]$  & $\mathcal{C}_{14}$ & $[30,28,2]$ \\ \hline

\end{tabular}
\end{minipage}}
\end{table}

\end{enumerate}

\subsection{Group Reversible LCD Codes}

In this section we search for group reversible LCD codes using Theorem~\ref{RevLCDThm1}. In particular, we employ the matrices $\sigma(v)$ given in Equation~(\ref{dihedral1a}) and Equation~(\ref{dihedral1b}) but with the listing of the group elements in the group ring element $v$ as given in Theorem~\ref{RevLCDThm1}. That is, since the group used in the group ring element $v$ in both equations; Equation~(\ref{dihedral1a}) and Equation~(\ref{dihedral1b}), is the dihedral group of even order, we list the elements of this group (according to Theorem~\ref{RevLCDThm1}) as follow:
\begin{equation}\label{dihedralfixed}
\{e,a,a^2,\dots,a^{n-1},ba^{n-1},ba^{n-2},\dots,ba,b\}.
\end{equation}
This approach, as shown in Theorem~\ref{RevLCDThm1} guarantees that the codes are group reversible LCD codes.

\begin{enumerate}
\item[1.] Let $v \in \mathbb{F}_2D_{22}$ and let $\mathcal{G}_5=\langle \sigma(v) \rangle,$ where $\sigma(v)$ is the matrix given in Equation~(\ref{dihedral1b}). We list our results in Table~\ref{Table14}. 

\begin{table}[h!]
\caption{Group Reversible LCD Codes from $\mathcal{G}_5$ over $\mathbb{F}_2$}
\resizebox{1\textwidth}{!}{\begin{minipage}{\textwidth}\label{Table14}
\centering
\begin{tabular}{|c|c|c|c|}
\hline
$\mathcal{C}_i$ & $[n,k,d]$   & $\mathcal{C}_i$ & $[n,k,d]$   \\ \hline
$\mathcal{C}_1$ & $[22,2,11]$ & $\mathcal{C}_2$ & $[22,20,2]$ \\ \hline
\end{tabular}
\end{minipage}}
\end{table}

\item[2.] Let $v \in \mathbb{F}_2D_{24}$ and let $\mathcal{G}_6=\langle \sigma(v) \rangle,$ where $\sigma(v)$ is the matrix given in Equation~(\ref{dihedral1a}). We list our results in Table~\ref{Table15}.

\begin{table}[h!]
\caption{Group Reversible LCD Codes from $\mathcal{G}_6$ over $\mathbb{F}_2$}
\resizebox{1\textwidth}{!}{\begin{minipage}{\textwidth}\label{Table15}
\centering
\begin{tabular}{|c|c|c|c|}
\hline
$\mathcal{C}_i$ & $[n,k,d]$   & $\mathcal{C}_i$ & $[n,k,d]$   \\ \hline
$\mathcal{C}_1$ & $[24,8,3]$ & $\mathcal{C}_2$ & $[24,16,2]$ \\ \hline
\end{tabular}
\end{minipage}}
\end{table}

\item[3.] Let $v \in \mathbb{F}_2D_{30}$ and let $\mathcal{G}_7=\langle \sigma(v) \rangle,$ where $\sigma(v)$ is the matrix given in Equation~(\ref{dihedral1b}). We list our results in Table~\ref{Table16}.

\begin{table}[h!]
\caption{Group Reversible LCD Codes from $\mathcal{G}_7$ over $\mathbb{F}_2$}
\resizebox{1\textwidth}{!}{\begin{minipage}{\textwidth}\label{Table16}
\centering
\begin{tabular}{|c|c|c|c|}
\hline
$\mathcal{C}_i$ & $[n,k,d]$   & $\mathcal{C}_i$ & $[n,k,d]$   \\ \hline
$\mathcal{C}_1$ & $[30,2,15]$ & $\mathcal{C}_8$ & $[30,16,4]$ \\ \hline
$\mathcal{C}_2$ & $[30,4,10]$ & $\mathcal{C}_{9}$ & $[30,18,4]$ \\ \hline
$\mathcal{C}_3$ & $[30,6,5]$  & $\mathcal{C}_{10}$ & $[30,20,2]$ \\ \hline
$\mathcal{C}_4$ & $[30,8,6]$  & $\mathcal{C}_{11}$ & $[30,22,2]$ \\ \hline
$\mathcal{C}_5$ & $[30,10,3]$ & $\mathcal{C}_{12}$ & $[30,24,2]$ \\ \hline
$\mathcal{C}_6$ & $[30,12,6]$ & $\mathcal{C}_{13}$ & $[30,26,2]$ \\ \hline
$\mathcal{C}_7$ & $[30,14,3]$  & $\mathcal{C}_{14}$ & $[30,28,2]$ \\ \hline
\end{tabular}
\end{minipage}}
\end{table}

\end{enumerate}

\section{Conclusion}

In this work, we presented a new method for constructing LCD codes over any finite commutative ring. We showed that the LCD codes constructed with our method are ideals in a group ring $RG,$ i.e., our LCD codes are also group codes. We proved that with a certain condition on the group ring element $v,$ one can construct non-trivial group LCD codes with our method. Moreover, we also proved that under some more restrictions on the group ring element $v,$ one can construct group reversible LCD codes with our method. We presented many examples of group LCD codes and group reversible LCD codes with different parameters. A possible direction for future research is to consider our approach but with groups of orders higher than we used in this paper to obtain more group LCD and group reversible LCD codes of lengths greater than we constructed in this paper.

\end{document}